\documentclass{article}
\usepackage[margin=1in]{geometry}
\usepackage{biblatex}
\usepackage{graphicx} 
\usepackage{amsmath,amsthm,amssymb,amsfonts}
\usepackage{amsmath}
\usepackage{bbm}
\newtheorem{proposition}{Proposition}[section]
\usepackage[usenames, dvipsnames,table]{xcolor}
\usepackage{hyperref}
\usepackage{bm}
\usepackage{marginnote}
\usepackage{setspace}
\usepackage{tikz}
\usepackage{longtable}
\usepackage{booktabs,tabularx,array}
\usepackage{arydshln}
\usetikzlibrary{positioning,arrows.meta}
\usepackage{authblk}

\title{Markov state models revisited: \\
Principles and algorithms for unbiased observables}
\author[1]{David Aristoff}
\affil[1]{Colorado State University, aristoff@rams.colostate.edu}
\author[2]{Robert J. Webber}
\affil[2]{University of California, San Diego, rwebber@ucsd.edu}
\author[3]{Daniel M. Zuckerman}
\affil[3]{Oregon Health and Science University, zuckermd@ohsu.edu}

\date{\today}

\newcommand{\cij}{C_{ij}}
\newcommand{\cijk}{C_{ij}^{(k)}}

\newcommand{\ci}{C_i}
\newcommand{\cik}{C_i^{(k)}}

\newcommand{\flux}{\Phi}
\newcommand{\fij}{\Phi_{ij}}
\newcommand{\fatobij}{\flux^{A \to B}_{ij}}

\newcommand{\ind}{I}
\newcommand{\kab}{k_{AB}}
\newcommand{\key}[1]{\vspace{3mm} \hrule \vspace{1mm} \noindent KEY POINT: #1 \vspace{1mm} \hrule \vspace{3mm}}

\newcommand{\pp}{{\pi}} %
\newcommand{\patob}{\pi^{A \to B}}
\newcommand{\pequil}{\pi^{\mathrm{equil}}}
\newcommand{\ppequil}{\pp^{\mathrm{equil}}}

\newcommand{\pimacroatob}{\tilde{\pi}^\mathrm{macro}}

\newcommand{\pimicroatob}{\tilde{\pi}^\mathrm{micro}}
\newcommand{\pinit}{\rho^{X}}
\newcommand{\px}{\pp^{X}}
\newcommand{\ppatob}{\pp^{A \to B}}

\newcommand{\ptau}{p_\tau}
\newcommand{\ptauatob}{\ptau^{A \to B}}
\newcommand{\ptotatob}{{\overline{\pp}}^{A \to B}}

\newcommand{\qa}{q^A}
\newcommand{\qb}{q^B}

\newcommand{\tatob}{T^{A \to B}}

\newcommand{\tij}{T_{ij}}
\newcommand{\tatobij}{T^{A \to B}_{ij}}
\newcommand{\tequil}{T^{\mathrm{equil}}}
\newcommand{\tequilij}{T^{\mathrm{equil}}_{ij}}
\newcommand{\tloc}{m} 
\newcommand{\tlocatob}{m^{A \to B}}

\newcommand{\txij}{{T_{ij}^{X}}}
\newcommand{\tx}{T^{X}}
\newcommand{\tmat}{{T}}

\newcommand{\tmacroatob}{\tilde{T}^\mathrm{macro}}

\newcommand{\tmicro}{T^\mathrm{micro}}
\newcommand{\tmicroab}{\tmicro_{\alpha \beta}}
\newcommand{\tmicroatob}{\tilde{T}^\mathrm{micro}}
\newcommand{\tmsm}{T^\mathrm{MSM}}
\newcommand{\tmsmatob}{\tilde{T}^\mathrm{MSM}}
\newcommand{\wtnew}{w^\mathrm{new}}

\begin{document}

\maketitle

\begin{abstract}
Markov state models (MSMs) have become ubiquitous tools for analyzing molecular dynamics (MD) simulations because of their simple, powerful premise:
although complete MD sampling may be impossible, the MSM can ``stitch together'' transition probabilities derived from local sampling to provide a global picture of kinetics and mechanisms.
In the standard MSM framework, the available MD data is organized into a single transition matrix, which is then used to estimate all observables at a lag time chosen so the coarse-grained dynamics are approximately Markovian.
This approach leads to avoidable model bias and motivates long lag times that obscure short-timescale processes of interest.
In contrast, this paper shows how to obtain unbiased coarse-grained observables at any fixed lag time and for any fixed coarse-graining in the limit of infinite, properly weighted data.
The central idea is to replace the single-matrix framework with two transition matrices --- one representing equilibrium dynamics and another representing source-sink recycling dynamics --- and use the correct matrix or matrices to estimate the matched dynamical observables.
\end{abstract}





\section{Introduction}



For more than twenty years, researchers have used Markov state models (MSMs) to analyze molecular dynamics (MD) simulation data \cite{swope2004msm-one,swope2004msm-two}.
From the beginning, their vision was to derive a single transition matrix, coarse-grained in space and discretized in time, which describes all the equilibrium, kinetic, and mechanistic features of a molecular system.
Researchers have acknowledged that MSMs are inherently approximate \cite{prinz2011markov,nuske2017markov,wu2017variational}, yet additional steps are needed to describe and address all sources of approximation error.
See Fig.\ \ref{fig:msm-pipeline-errors} for an overview of the MSM model building pipeline and its error sources.

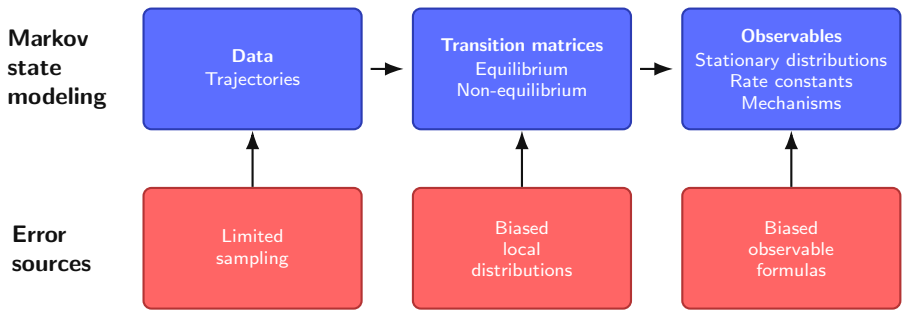
\begin{figure}[t]
\centering

\definecolor{msmblue}{HTML}{3F55FF}
\definecolor{msmred}{HTML}{FF4A4A}
\definecolor{msmgray}{HTML}{8F8F8F}
\definecolor{msmblack}{HTML}{111111}

\begin{tikzpicture}[
    >=Latex,
    node distance=1cm and .65cm,
    box/.style={
        thick,
        rounded corners=3pt,
        text width=2.6cm,
        minimum height=1.6cm,
        align=center,
        font=\sffamily\scriptsize,
        text=white,
        inner sep=4pt
    },
    topbox/.style={
        box,
        fill=msmblue!85!white,
        draw=msmblue!70!black
    },
    bottombox/.style={
        box,
        fill=msmred!85!white,
        draw=msmred!70!black
    },
    side label/.style={
        font=\sffamily\bfseries\small,
        text=msmblack,
        align=left
    },
    arrow/.style={
        thick,
        draw=msmblack,
        -{Latex[length=2.5mm,width=1.8mm]}
    }
]

\node[topbox] (data) {
    \textbf{Data}\\[1pt]
    Trajectories
};

\node[topbox, right=of data] (framework) {
    \textbf{Transition matrices}\\[1pt]
    Equilibrium\\
    Non-equilibrium
};

\node[topbox, right=of framework] (observables) {
    \textbf{Observables}\\[1pt]
    Stationary distributions\\
    Rate constants\\
    Mechanisms
};

\draw[arrow] ([xshift=3pt]data.east)
    -- ([xshift=-3pt]framework.west);

\draw[arrow] ([xshift=3pt]framework.east)
    -- ([xshift=-3pt]observables.west);

\node[bottombox, below=0.75cm of data] (sampling) {
    Limited\\
    sampling
};

\node[bottombox, below=0.75cm of framework] (estimators) {
    Biased\\
    local\\
    distributions
};

\node[bottombox, below=0.75cm of observables] (bias) {
    Biased\\
    observable\\
    formulas
};

\draw[arrow] (sampling.north) -- (data.south);
\draw[arrow] (estimators.north) -- (framework.south);
\draw[arrow] (bias.north) -- (observables.south);

\node[side label, left=.35cm of data] (msm) {
    Markov\\
    state\\
    modeling
};

\node[side label, left=.55cm of sampling] (errors) {
    Error\\
    sources
};

\end{tikzpicture}

\caption{Markov state modeling pipeline and the corresponding sources of error.
This perspective focuses on two systematic errors: errors from biased within-cluster sampling distributions and biased observable formulas.}
\label{fig:msm-pipeline-errors}

\end{figure}

The goal of this perspective paper is to clearly set out principles and algorithms for ``unbiased'' estimates of both equilibrium and nonequilibrium observables.
Here, an ``unbiased'' estimate is one that converges to a theoretically exact, coarse-grained observable in the limit of infinite, properly weighted data.
We argue that the current literature has overlooked some of the capacity for MSMs to provide unbiased estimates.
We will explain what to do and what \emph{not} to do when using MSMs.
Along the way, we will critically re-examine MSM approaches and emphasize three related themes.
\begin{enumerate}
\item \emph{Markovian dynamics do not lead to coarse-scale Markovian behavior.} Traditional MSMs take literally the Markov assumption by using a single transition matrix to compute all observables.
However, a coarse-grained system does not typically satisfy the Markov property at short lag times, leading to biased dynamical estimates.
\item \emph{Unbiased estimation requires two transition matrices.} By using two transition matrices that represent the equilibrium and nonequilibrium steady-state (NESS) conditions \cite{russo2021unbiased}, we can provide unbiased estimates of dynamical observables for any lag time and choice of coarse-graining.
In this nontraditional approach, each observable can be estimated without bias from the matched transition matrix or matrices.
\item \emph{Unbiased estimation requires stationary local distributions.} 
MSM estimates are sensitive to the local distributions within clusters \cite{wu2017variational}, which must be stationary with respect to equilibrium or source-sink recycling dynamics for unbiased estimation of observables.
While traditional MSM computations are locked into the data at hand, the newly introduced RiteWeight algorithm is designed to reweight trajectories toward the correct equilibrium or nonequilibrium distribution \cite{kania2026randomized}.
\end{enumerate}
The rest of this introductory section will describe the three themes in more detail (Secs.~\ref{sec:not-markov}--\ref{sec:local-global}) and then compare our approach against previous MSM analyses (Sec.~\ref{sec:comparison}).

\subsection{Markovian dynamics do not lead to coarse-scale Markovian behavior}
\label{sec:not-markov}

The heart of the MSM approach is the transition matrix $\tmat$.
The elements of this matrix,
\begin{equation*}
\tij = \tij(\tau),
\end{equation*}
specify the conditional transition probabilities between clusters or MSM ``states''.
Given a start point drawn from a specified distribution within
cluster $i$, the element $\tij$ is the conditional probability of
being found in cluster $j$ one lag time later.
The transition probability only accounts for the location of the trajectory after the full lag time, without accounting for visits to $j$ before the end of the interval.

The MSM construction leads to a subtle point of confusion regarding the Markovian nature of the dynamics.
In molecular simulation, the vast majority of computer algorithms that propagate dynamics are Markovian, since the current state is sufficient to generate the next state of a system.
If a simulation algorithm uses auxiliary variables or a finite history,
the state can be enlarged to include those variables so that the resulting
dynamics are Markovian.
However, once phase space points are grouped together into clusters, the Markov property no longer holds at short lag times \cite{prinz2011markov}.

We can understand the failure of Markovianity by considering the energy landscapes of biomacromolecules (Fig.\ \ref{fig:landscape}).
Full proteins contain hundreds or even thousands of residues, and even a simple tripeptide has $>100$ energy basins \cite{elber1990rxnpaths}.
Due to the finite availability of data and the complexity of configuration space, clusters must contain multiple energy basins \cite{elber1990rxnpaths,wales2018landscapes,wales2019pathways}.
Therefore, a trajectory's specific location within a cluster affects the future probability to be in other clusters. 

\begin{figure}[t]
    \centering
    \includegraphics[width=.95\linewidth]{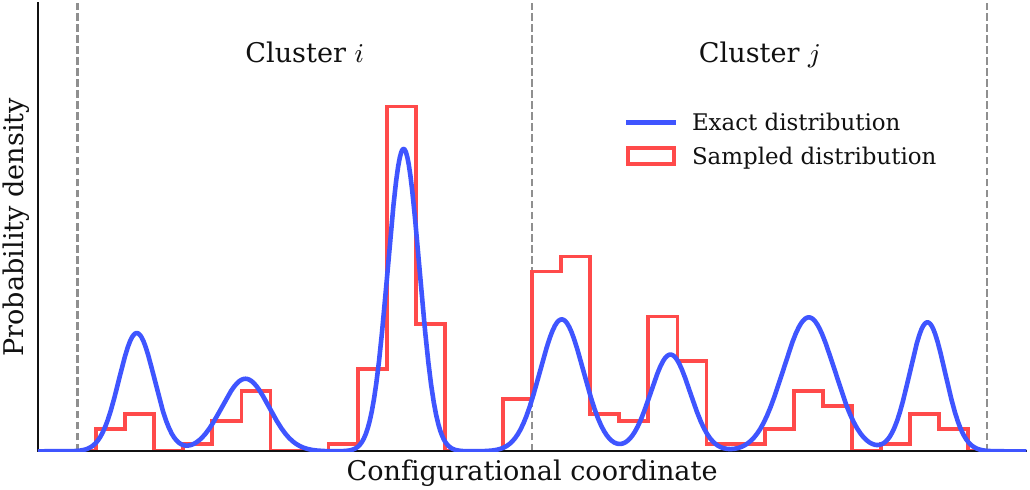}
    \caption{Schematic rough energy landscape and clusters of an MSM.  A biomacromolecule consists of thousands of atoms, so individual clusters $i$ and $j$ of the Markov state model inevitably contain many local energy basins.
    The finite available trajectory data (red histogram) will not typically match the targeted exact distribution internal to the clusters (blue line).}
    \label{fig:landscape}
\end{figure}

Traditional MSMs employ long lag times to improve Markovian behavior and enhance the accuracy of observables.
As the lag time grows, the original phase-space point or ``microstate'' of the trajectory becomes less determinative of the final cluster, so Markovian behavior is approached \cite{chodera-noe2014markov-review}. 
However, the Markovian behavior comes at a cost.
A long lag time can degrade an MSM's ability to provide mechanistic insights, since longer lag times jump over behaviors occurring at shorter timescales \cite{suarez2021markov-can-not}.

The challenge of MSM modeling can thus be understood in terms of two intrinsic timescales for each cluster $i$.
\begin{itemize}
    \item $t_{\rm relax}^i$ is the timescale for a trajectory to relax to local equilibrium within cluster $i$.
    \item $t_{\rm exit}^{i}$ is the timescale for the trajectory to exit cluster $i$, subsequent to relaxation.
\end{itemize}
Ideally, a single lag time $\tau$ is chosen such that for all clusters $i$,
\begin{equation}
    t_{\rm relax}^i \ll \tau \ll t_{\rm exit}^i.
\end{equation}
Violations of the first inequality break Markovianity, while violations of the second inequality miss events that may be of physical interest.
The two requirements may be impossible to satisfy simultaneously
with a single lag time.

\subsection{Unbiased estimation requires two transition matrices} \label{sec:two_matrices}

A fundamental issue is that different types of observables correspond naturally to different ensembles. 
Equilibrium state probabilities are naturally associated with the
equilibrium ensemble, whereas rate constants and fluxes can be
represented conveniently through source--sink nonequilibrium ensembles \cite{vanden2009tilting,dinner2009separating,suarez2014simultaneous,suarez2021markov-can-not}.
In most cases, qualitative differences between the ensembles cannot be captured by a single transition matrix \cite{bhatt2011reversibility}.
To eliminate the bias in traditional MSMs, we will use two ensemble-specific transition matrices, one representing equilibrium conditions and one representing source-sink conditions.
Tab.~\ref{tab:ab_observables} identifies which matrix or matrices are needed for unbiased estimation of each observable.
We will present detailed formulas for unbiased estimators in the later sections.

\begin{table}[t]
\centering
\begingroup
\setlength{\tabcolsep}{4pt}
\begin{tabular}{l|ccccc}
\hline
Required distributions 
& $\ppequil_i$ 
& $\ppatob_i$ 
& Committor $\qa_i$ 
& MFPT $\tlocatob$
& Flux $\fatobij$ \\
\hline
\noalign{\vskip 3pt}
Equilibrium 
& $\checkmark$ 
& 
& $\checkmark$ 
& 
&  \\
$A \to B$ NESS 
& 
& $\checkmark$ 
& $\checkmark$ 
& $\checkmark$ 
& $\checkmark$ \\
\hline
\end{tabular}
\endgroup
\caption{Key observables and the distributions required for their unbiased estimation. 
Here $\ppequil_i$ is the equilibrium stationary probability for cluster $i$, 
$\ppatob_i$ is the stationary probability for cluster $i$ in the $A \to B$ nonequilibrium steady state (NESS),
$\qa_i$ is the committor representing the probability of reaching state $A$ before $B$,
$\tlocatob$ is the mean first-passage time (MFPT) from state $A$ to state $B$ starting from cluster $i$,
and $\fatobij$ is the $A \to B$ probability flux from state $i$ into state $j$ over one lag interval.}
\label{tab:ab_observables}
\end{table}

\subsection{Unbiased estimation requires stationary local distributions}
\label{sec:local-global}

To estimate unbiased observables, we need local sampling within each cluster to match the appropriate equilibrium or non-equilibrium ensemble as described in Tab.~\ref{tab:ab_observables}.
The key to correct local sampling is  the principle of \emph{stationarity}.
Both the equilibrium and $A \to B$ steady state distributions are stationary distributions, after accounting for the appropriate recycling boundary conditions.
This means the amount of probability mass flowing into each region matches the amount of probability mass leaving each region.
Remarkably, as we show below, we can estimate stationary distributions for a coarse-grained system using only the microscopic Markov property and the principle of stationarity.
No Markov property for the coarse-scale model is required.
In our discussion of computational approaches (Sec. \ref{sec:algo}), we will emphasize the newly introduced RiteWeight algorithm \cite{kania2026randomized}, which reweights trajectories in a self-consistent iteration designed to make their weighted distribution consistent with the target equilibrium or nonequilibrium stationary distribution.

\subsection{Comparison to previous analyses} \label{sec:comparison}

For general background on MSMs, the review paper \cite{husic2018markov} provides historical perspective, and the article \cite{pande2010everything} gives a user-friendly introduction.
Two Python packages that made MSMs accessible to a wide audience are MSMBuilder \cite{harrigan2017msmbuilder} and PyEMMA \cite{schererpyemma2015}. While these packages are no longer actively maintained, many of their capabilities are now available through the deeptime package \cite{hoffmann2021deeptime}.

Beyond general introductions,
several works analyze the error in MSM observables and discuss ways to reduce the error.
Sarich, No\'e, and Sch\"utte \cite{sarich2010approximation} and Prinz et al. \cite{prinz2011markov} analyze projection errors and show how the coarse-grained MSM dynamics fail to satisfy the Markov property.
N\"uske et al. \cite{nuske2017markov} analyze a different error source due to mismatched sampling distributions within MSM states.
These early analyses motivate longer lag times, improved state definitions, and practical estimators for initial-distribution bias. 

The recommendations for improving MSMs are valuable, but they differ from our emphasis in the present work.
In many applications, constructing states that accurately resolve the slowest dynamical modes is difficult, and increasing the lag time skips over short-timescale processes of biological significance.
We therefore focus on unbiased estimation of exact coarse-grained observables for any fixed lag time and partition.
In particular, we advocate using two transition matrices based on equilibrium and nonequilibrium steady-state distributions, as well as reweighting trajectories according to the RiteWeight algorithm \cite{kania2026randomized}.

Another relevant body of work considers dynamical observables derived from MSMs, such as the committors, mean first-passage times, and nonequilibrium fluxes in Tab.~\ref{tab:ab_observables}.
Su{\'a}rez et al.\ identify which coarse-grained observables can be estimated without bias from traditional MSMs \cite{suarez2021markov-can-not}.
Related work by Russo et al. \cite{russo2020iterative} constructs unbiased estimators for dynamical observables.
The present paper builds on the previous development by identifying which ensemble-specific transition matrices are required for unbiased estimation of each observable.
Of equal importance, we present all the necessary theory in a consistent and accessible fashion.

Finally, past literature has studied the impact of finite-sampling error in MSMs.
Singhal and Pande \cite{singhal2005error} pointed attention to the propagation of sampling error in estimates of mean first-passage times.
Noe \cite{noe2008probability} and Chodera and No\'e \cite{chodera2010bayesian} developed Bayesian approaches for estimating and addressing sampling error, and Trendelkamp-Schroer et al. \cite{trendelkamp2015estimation} extended this framework to time-reversible MSMs.
Kozlowski and Grubmuller compared error sources in MSMs and emphasized the role of finite sampling \cite{grubmuller2023uncertainties}.
Most recently, Tuchkov et al. \cite{tuchkov2025error}
analyzed mean first-passage time and committor estimates from MSMs, including their sensitivity to finite sampling.
These works are complementary to the present paper.
They address uncertainty caused by finite data, whereas we focus on systematic bias that remains in the infinite-data limit when transition matrices are constructed from mismatched within-cluster distributions or
when observable formulas incorrectly assume coarse-grained Markovianity.

\section{A fully solvable discrete model}
\label{sec:8-state-model}

To gain insight into the fundamental limitations of MSMs, we start by examining a discretized model of a one-dimensional diffusion process.
For concreteness, we assume eight microstates are a perfect physical description of our system; there are no finer states.
The ground-truth transition probabilities are defined for a particular lag time $\tau$ by the microscopic transition matrix,
\begin{equation}
\tmicro =
\begin{pmatrix}
2/3 & 1/3 & 0   & 0 & 0 & 0 & 0 & 0 \\
1/3 & 1/3 & 1/3 & 0 & 0 & 0 & 0 & 0 \\
0   & 1/3 & 1/3 & 1/3 & 0 & 0 & 0 & 0 \\
0   & 0   & 1/3 & 1/3 & 1/3 & 0 & 0 & 0 \\
0   & 0   & 0   & 1/3 & 1/3 & 1/3 & 0 & 0 \\
0   & 0   & 0   & 0   & 1/3 & 1/3 & 1/3 & 0 \\
0   & 0   & 0   & 0   & 0 & 1/3 & 1/3 & 1/3 \\
0   & 0   & 0   & 0   & 0 & 0   & 1/3 & 2/3
\end{pmatrix}
\label{tmicro}
\end{equation}
Each element $\tmicroab$ of this matrix represents the probability to transition from microstate $\alpha$ to $\beta$ in a lag time $\tau$, or more precisely the probability for a trajectory to be observed at $\beta$ after having been at $\alpha$ a time increment $\tau$ earlier.
From each state there is a $1/3$ probability to move one step left or right or to remain in the same state, except for the extreme left- and right-most states where there is a $2/3$ probability to remain in the state.

The microscopic transition matrix $\tmicro$ completely defines our system.
Thus, the exact dynamics are described by
\begin{equation}
\label{eq:probability_flux}
    \pp_\beta(t+\tau) =  \sum_\alpha \pp_\alpha(t) \, \tmicroab,
\end{equation}
where $\pp_\alpha(t)$ denotes the normalized probability for each microstate $\alpha$ at time $t$.
The probability observed in microstate $\beta$ is a sum over all probability flowing into that state or remaining in the state. 
Note that summation over $\alpha$ includes $\alpha=\beta$, and the term $\pp_\beta(t) \, \tmicro_{\beta \beta}$ implicitly captures probability remaining in the state.

From equation \eqref{eq:probability_flux}, the equilibrium solution for the matrix $\tmicro$ is simply a uniform distribution over the eight states, as expected.
\begin{equation}
    \pequil_\alpha=\frac{1}{8} \hspace{1cm} (\alpha = 1, 2, \ldots, 8).
    \label{pequil-eighth}
\end{equation}
We can confirm the detailed balance condition holds for this distribution,
\begin{equation}
    \pequil_\alpha \, \tmicroab = \pequil_\beta \tmicro_{\beta \alpha}.
    \label{equil-micro}
\end{equation}
Because of detailed balance, the equilibrium distribution is stationary, and it has no probability sources or sinks.

\subsection{Equilibrium comparison between coarse- and fine-grained distributions}

In an MSM, we organize the microstates into groups and compute an associated  coarse-grained transition matrix.
We do this now, not by using trajectory data as would be usual \cite{chodera-noe2014markov-review}, but instead by calculating the idealized MSM that would be obtained from an infinite trajectory generated using $\tmicro$. In later examples and discussion, we will examine the consequences of finite data in MSMs.

If we group microstates into neighboring pairs, the eight microstates become four clusters of two microstates each.
We can compute the exact $4 \times 4$ MSM equilibrium transition matrix $\tmsm$ by averaging transition probabilities weighted by the equilibrium distribution.
The elements of this matrix are
\begin{equation}
    \tmsm_{ij} 
    = \frac { \sum_{\alpha \in i} \sum_{\beta \in j} \pequil_\alpha \tmicroab }
            { \sum_{\alpha \in i} \pequil_\alpha }
    = \sum_{\alpha \in i} \sum_{\beta \in j} \frac{1}{2} \, \tmicroab
    \label{tmacroij}
\end{equation}
where the notation $\alpha \in i$ refers to all microstates in the coarse state, e.g., $\alpha = 1, 2$ for $i = 1$.
In \eqref{tmacroij}, the first equality is general and the second is particular to the 8-state model under consideration. 
These calculations yield the same matrix $\tmsm_{ij}$
that would result from counting transitions in an infinite equilibrium-distributed trajectory.

The defining equation for the macroscopic matrix elements \eqref{tmacroij} leads to
\begin{equation}
\tmsm =
\begin{pmatrix}
5/6 & 1/6 & 0 & 0 \\
1/6 & 2/3 & 1/6 & 0 \\
0 & 1/6 & 2/3 & 1/6 \\
0 & 0 & 1/6 & 5/6
\end{pmatrix}.
\label{tmacro-equil}
\end{equation}
This coarse-grained transition matrix allows us to compute equilibrium quantities of interest, and we start with the simplest: equilibrium populations.  
Using the coarse-grained analog of the detailed balance condition \eqref{equil-micro}, namely,
\begin{equation}
    \pequil_i \, \tmsm_{ij} = \pequil_j \tmsm_{ji},
    \label{equil-macro}
\end{equation}
we can verify that the expected solution $\pi_i=1/4$ holds for every $i$.
The uniform coarse-grained solution is correct because we obtain the same value $1/8 + 1/8 = 1/4$ when we sum the microscopic equilibrium probabilities in each coarse-grained state.
The coarse-grained and microscopic solutions are consistent.

\subsection{Nonequilibrium comparison between coarse- and fine-grained distributions}

The consistency between microscopic and coarse-grained MSM predictions can break down if we impose nonequilibrium boundary
conditions after the
coarse-graining. Here, we will focus on a source-sink steady state distribution with the source (initial) state $A$ corresponding to the first coarse state ($i=1$) and the sink (target) state $B$ corresponding to the final coarse state ($i=4$).  
Trajectories reaching the sink state are re-started in the source state at the next interval of the lag time, according to a modified microscopic matrix
\begin{equation}
\tilde{T}^{\rm micro} =
\begin{pmatrix}
2/3 & 1/3 & 0   & 0 & 0 & 0 & 0 & 0 \\
1/3 & 1/3 & 1/3 & 0 & 0 & 0 & 0 & 0 \\
0   & 1/3 & 1/3 & 1/3 & 0 & 0 & 0 & 0 \\
0   & 0   & 1/3 & 1/3 & 1/3 & 0 & 0 & 0 \\
0   & 0   & 0   & 1/3 & 1/3 & 1/3 & 0 & 0 \\
0   & 0   & 0   & 0   & 1/3 & 1/3 & 1/3 & 0 \\
0 & 1 & 0 & 0 & 0 & 0 & 0 & 0 \\
0 & 1 & 0 & 0 & 0 & 0 & 0 & 0 \\
\end{pmatrix}.
\label{tmicro-atob}
\end{equation}
The specific choice of feedback to the $\alpha = 2$ state from the $\alpha = 7, 8$ states is justified by the first entry distribution discussed later in Sec.\ \ref{sec:nonequilibrium}, which builds in a natural correspondence with the equilibrium ensemble \cite{bhatt2011reversibility}.
The $A \to B$ nonequilibrium steady state (NESS) is defined by the leading left eigenvector of the microscopic matrix \eqref{tmicro-atob}, which is
\begin{equation}
    \tilde{\pi}^{\rm micro} = \frac{1}{61} \begin{pmatrix}
        15 \\ 15 \\ 12 \\ 9 \\ 6 \\ 3 \\ 1 \\ 0
    \end{pmatrix}
    \label{pimicroatob}
\end{equation}
Hence, $\tilde{\pi}^{\rm micro}$ is the unique stationary vector for the $A \to B$ recycling dynamics.
Due to the recycling, the first entry into the sink $B$ occurs at microstate $\alpha=7$, and the process is recycled before moving to $\alpha=8$; hence the stationary probability of $\alpha=8$ is zero.

In the unbiased approach to coarse-graining, we define the $A \to B$ coarse-grained transition matrix via stationary weighted averages of the microscopic transition probabilities,
\begin{equation}
\label{eq:tmacroatobij}
    \tmacroatob_{ij} 
    = \frac { \sum_{\alpha \in i} \sum_{\beta \in j} \pimicroatob_\alpha \,\tmicroatob_{\alpha \beta} }
    { \sum_{\alpha \in i} \pimicroatob_\alpha}.
\end{equation}
The defining equation \eqref{eq:tmacroatobij} yields the NESS coarse-grained matrix
\begin{equation}
\tmacroatob =
\begin{pmatrix}
5/6 & 1/6 & 0 & 0 \\
4/21 & 14/21 & 3/21 & 0 \\
0 & 2/9 & 6/9 & 1/9 \\
1 & 0 & 0 & 0
\end{pmatrix}
\label{tmacroatob}
\end{equation}
In the rows $i = 2, 3, 4$, the NESS coarse-grained matrix $\tmacroatob$ differs from the equilibrium coarse-grained matrix \eqref{tmacro-equil}.

In contrast, the traditional, biased route first constructs the equilibrium coarse-grained matrix and only afterward imposes the source--sink boundary condition. This procedure keeps the equilibrium-weighted transition rows for the intermediate coarse states, rather than recomputing those rows using the $A \to B$ NESS distribution.
This yields a coarse-grained $A \to B$ MSM transition matrix,
\begin{equation}
\tmsmatob =
\begin{pmatrix}
5/6 & 1/6 & 0 & 0 \\
1/6 & 2/3 & 1/6 & 0 \\
0 & 1/6 & 2/3 & 1/6 \\
1 & 0 & 0 & 0
\end{pmatrix}
\label{tmsmatob}
\end{equation}
Compared to the original MSM matrix \eqref{tmacro-equil}, probability arriving at the fourth state ($B$ = sink) is fed back to the first state ($A$ = source).
This MSM matrix differs from the NESS matrix \eqref{tmacroatob} because probability flows symmetrically from the intermediate states ($i = 2, 3$).
In the NESS matrix, probability is tilted to flow to the left, in the direction of the source state.

With different transition matrices from the two coarse-graining schemes, not surprisingly the coarse-grained stationary vectors for \eqref{tmacroatob} and \eqref{tmsmatob} differ.
\begin{equation}
    \pimacroatob = \frac{1}{61} \begin{pmatrix}
        30 \\ 21 \\ 9 \\ 1
    \end{pmatrix}
    \approx
    \begin{pmatrix}
        0.492 \\ 0.344 \\ 0.148 \\ 0.016
    \end{pmatrix}
    \quad \text{and} \quad
    \tilde{\pi}^{\rm MSM} 
    = \frac{1}{37} \begin{pmatrix}
        18 \\ 12 \\ 6 \\ 1
    \end{pmatrix}
    \approx \begin{pmatrix}
        0.486 \\ 0.324 \\ 0.162 \\ 0.027
    \end{pmatrix}.
    \label{8state-pimacro}
\end{equation}
Only the stationary probability vector $\pimacroatob$ is correct because it matches the coarse-graining of the microscopic stationary vector \eqref{pimicroatob}, i.e., summing the elements in pairs.
The MSM vector slightly underweights the first two states and overweights the last two states.

Are these differences important?
The lesson here is not the large or small numeric differences, because the magnitudes could be tuned by adjusting the underlying microscopic transition matrix. Rather, the message is the following.

\key{
Even with infinite equilibrium sampling, the traditional MSM route to computing nonequilibrium observables is theoretically flawed, leading to biased observables.}

We will elaborate on this key point in the remainder of the perspective paper, demonstrating that the bias observed in the stationary vectors \eqref{8state-pimacro} is mirrored in biased mean first-passage times, committors, and mechanistic path fluxes.

\section{Dynamical observables for a continuous-state system}

In this section, we define equilibrium and dynamical observables for a continuous-state system.
For notation, we let $x$ denote a phase-space variable that represents all positions, velocities, and other degrees of freedom needed to make the dynamics Markovian.
Trajectories $x(t)$ are discretized by a lag time $\tau$, so they are observed at times
\begin{equation}
    x(0), \, x(\tau), \, x(2\tau), \, \ldots \,.
    \label{xtraj-tau}
\end{equation}
The trajectories evolve according to a transition probability density
\begin{equation}
    \ptau(x|x'),
    \label{ptau}
\end{equation}
which describes the dynamics initiated from $x'$ and observed after a time interval $\tau$.
For example, the transition probability density might represent MD simulation with a thermostat or barostat.
Because the system is Markovian in the full phase space, only a single prior phase point $x'$ is needed to specify the distribution of outcomes.
See Tab.~\ref{tab:notation} for a list of notation used throughout the paper.

\begin{table}[t]
\centering
\small
\renewcommand{\arraystretch}{1.06}
\setlength{\tabcolsep}{5pt}

\begin{tabular}{
    @{}
    >{\raggedright\arraybackslash}p{0.15\linewidth}
    >{\raggedright\arraybackslash}p{0.7\linewidth}
    @{}
}
\toprule
\textbf{Symbol} & \textbf{Description} \\
\midrule

$x$, $x'$ 
& Microscopic phase-space points \\

$i,j$ 
& Cluster labels \\

$k$ 
& Trajectory label \\

$\tau$ 
& Lag time used throughout the analysis \\

$\ptau(x\mid x')$ 
& Underlying transition probability density \\

$\ptau^{A\to B}(x\mid x')$ 
& Transition probability density for the $A\to B$ ensemble \\

Superscript $X$ 
& Ensemble label: equilibrium, $A\to B$, or $B\to A$ \\

$\px(x)$ 
& Stationary probability density in ensemble $X$ \\

$\px(x\mid i)$ 
& Stationary density in ensemble $X$, conditioned on cluster $i$ \\

$\px_i$ 
& Stationary probability of cluster $i$ in ensemble $X$ \\

$\px$ 
& Vector of cluster stationary probabilities in ensemble $X$ \\

$\flux_{ij}^X$ 
& Probability flux from cluster $i$ to cluster $j$ over one lag interval \\

$\tlocatob$ 
& Mean first-passage time from $A$ to $B$ \\

$\qa(x)$ 
& Microscopic probability of reaching $A$ before $B$ \\

$\qa_i$ 
& Equilibrium-averaged committor to $A$ within cluster $i$ \\

$\qb(x)$ 
& Microscopic probability of reaching $B$ before $A$ \\

$\cij^{(k)}$ 
& Number of observed $i\to j$ transitions for trajectory $k$ \\

$\ci^{(k)}$ 
& Total number of transitions originating in cluster $i$ for trajectory $k$ \\

$\pinit(x\mid i)$ 
& Distribution of trajectory initial points within cluster $i$ \\

$T_{ij}$ 
& Cluster transition matrix, for example estimated from transition counts \\

$\txij$ 
& Ideal cluster transition matrix in ensemble $X$ \\

$\ind_i(x)$ 
& Indicator function of cluster $i$ \\

$\ptotatob$ 
& Normalization constant relating the equilibrium and $A\to B$ densities \\

$\tloc_i^B$ 
& Local mean first-passage time from cluster $i$ to target state $B$ \\

\bottomrule
\end{tabular}

\caption{Notation used throughout the paper.}
\label{tab:notation}
\end{table}


The following subsections introduce equilibrium and nonequilibrium stationary densities for the continuous-state system (Secs.~\ref{sec:equilibrium} and \ref{sec:nonequilibrium}).
We then define dynamical observables for the system, including the committor and mean first-passage time (Secs.~\ref{sec:continuous_committor} and \ref{sec:mfpt-micro}).

\subsection{Equilibrium density} \label{sec:equilibrium}

We assume the continuous-state system has an equilibrium density $\pequil$ that satisfies detailed balance,
\begin{equation}
    \pequil(x') \, \ptau(x|x') = \pequil(x) \, \ptau(x'|x),
    \label{det-bal}
\end{equation}
or the corresponding momentum-flip version of detailed balance if $x$ includes momenta or other variables that are odd under time reversal \cite{leimkuhler2015molecular} .
The equilibrium density is stationary with respect to the transition density $p_{\tau}$ because it satisfies the relation
\begin{equation}
    \int \pequil(x') \, \ptau(x|x') \, \mathrm{d}x' = \pequil(x).
    \label{pss-condn}
\end{equation}
Detailed balance implies stationarity, as we can see by integrating both sides of \eqref{det-bal} over $x'$, but the converse is not true. 
When a stationary distribution does not satisfy detailed balance, it is called a ``nonequilibrium'' steady state distribution (NESS).

\subsection{Nonequilibrium steady-state densities} \label{sec:nonequilibrium}

\begin{figure}[t]
    \centering
    \includegraphics[width=\linewidth]{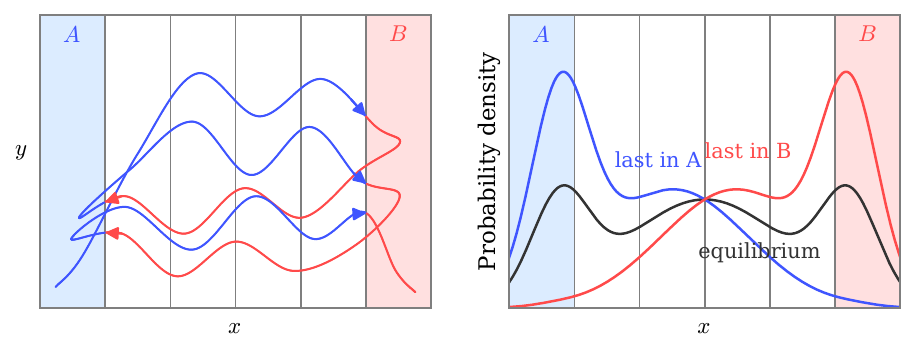}
    \caption{
    Decomposition of an equilibrium trajectory ensemble into two nonequilibrium steady-state ensembles.
    States $A$ and $B$ are the leftmost and rightmost boxes, and the vertical rectangles represent MSM clusters.
    Left: a long equilibrium trajectory is colored according to the most recently visited state.
    Blue segments belong to the $A \to B$ ensemble, while red segments belong to the $B \to A$ ensemble.
    Right: normalized probability densities for the equilibrium ensemble and the two directional ensembles. 
    }
    \label{fig:equil-decomp}
\end{figure}


Imagine a single trajectory is observed over a long time interval, long enough to contain many transitions between states $A$ and $B$.
The trajectory is in equilibrium throughout this extended period.
The trajectory can be traced backward in time until it reaches one of the states $A$ or $B$, and the trajectory is thus assigned to one of two classes, ``last in $A$'' or ``last in $B$'' (blue and red paths in Fig.\ \ref{fig:equil-decomp}).
Going forward in time, when the trajectory reaches the opposite state from its class, it changes class.

From this decomposition, we obtain two sets of trajectory segments, one for the $A \to B$ transitions and one for the $B \to A$ transitions.
In the $A \to B$ trajectory set, the segments all start at the source state $A$ and they all end at the sink state $B$.
The distribution of start points in $A$ is called the ``reactive entrance distribution''
\cite{lelievre2023hill-reactive} or the ``EqSurf'' distribution \cite{bhatt2011reversibility},
which we express
using the probability density $\rho_A(x)$.
As a convention, this paper chooses to include the endpoints in the $A \to B$ trajectory segments in order to simplify the description of the $A \rightarrow B$ recycling dynamics.

\begin{figure}
    \centering
    \includegraphics[width=.95\linewidth]{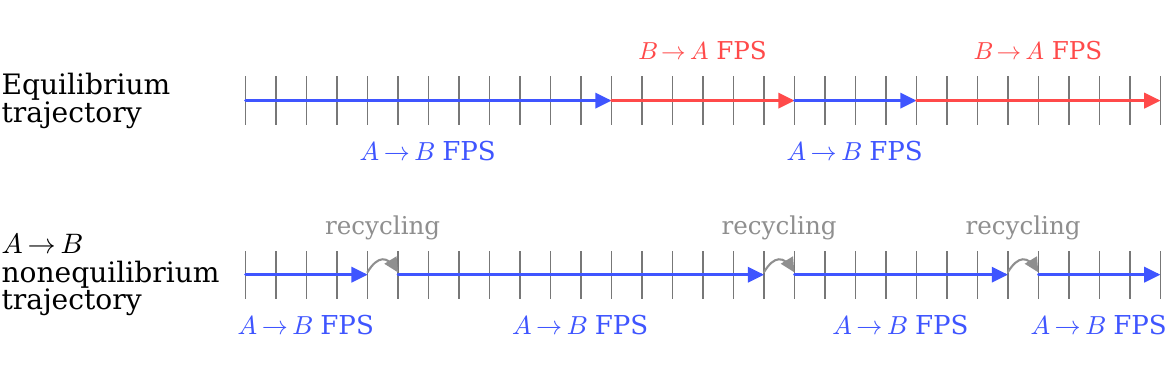}
    \caption{Long trajectories in the equilibrium and $A \to B$ nonequilibrium steady state ensembles.
    Both trajectories contain a sequence of $A \to B$ first passage segments (blue arrows).
    These are interrupted by $B \to A$ first passage segments (red arrows) in the equilibrium trajectory, or recycling segments (gray arrows) in the $A \to B$ nonequilibrium trajectory.  When $B \to A$ recycling is performed via the EqSurf distribution $\rho_A$, the distributions of the $A \to B$ segments are identical in the two ensembles.}
    \label{fig:long-traj}
\end{figure}

After extracting an infinite sequence of $A \to B$ trajectory segments, we randomly shuffle 
the segments and concatenate them to form a trajectory representing the $A \to B$ nonequilibrium ensemble, as illustrated in Fig.~\ref{fig:long-traj}.
In this ensemble, the transition probability density from one phase-space point $x'$ to the next phase-space point $x$ is given by the $A \to B$ recycling dynamics,
\begin{equation}
    \ptauatob(x|x') = \begin{cases} \ptau(x|x'), & x' \notin B, \\
    \rho_A(x), & x' \in B.
    \end{cases}
    \label{ptau-ab}
\end{equation}
The ensemble follows ordinary dynamics starting from any point $x' \notin B$, but from any $x' \in B$, the process is reset according to the EqSurf distribution $\rho_A$.

Last, we write $\patob(x)$ for the probability density function of phase-space points in the $A \to B$ nonequilibrium ensemble.
This density function is stationary with respect to the $A \to B$ recycling dynamics, namely,
\begin{equation}
    \int \patob(x') \, \ptauatob(x|x') \, \mathrm{d}x' = \patob(x).
\end{equation}
However, $\patob$ is not an equilibrium density because of the nonzero probability current.
Ordinary dynamics carry probability from $A$ to $B$, and the recycling dynamics return probability from $B$ to $A$.

\key{
The equilibrium ensemble decomposes into two nonequilibrium steady-state (NESS) ensembles, one for the $A \to B$ direction and another for $B \to A$ \cite{bhatt2011reversibility}.}

\subsection{Committor relates equilibrium to NESS} \label{sec:continuous_committor}

The committor $\qa(x)$ is the probability that trajectories starting from a phase-space point $x$ will visit state $A$ before they visit state $B$ \cite{gardiner1985book,vankampen2007book}.  Likewise $\qb$ is the analogous function in the opposite direction.  Since all trajectories eventually hit either $A$ or $B$, 
\begin{equation}
\qa(x) + \qb(x) = 1 
\label{q-symm}
\end{equation}
for every $x$.

We can relate the committor to the nonequilibrium steady-state density by considering a single long trajectory conforming to equilibrium.
Conditioned on the equilibrium trajectory passing through a phase-space point $x$, the probability that the forward trajectory next visits $A$ before $B$ is $\qa(x)$. 
Assuming detailed balance, we can reverse the trajectory in time to obtain a different sample from the equilibrium ensemble.
The trajectory segments next visiting state $A$ in the original trajectory correspond to trajectory segments that were last in $A$ in the reversed trajectory.
Therefore, the equilibrium density $\pequil(x)$ times the committor $\qa(x)$ is proportional to the nonequilibrium density $\patob(x)$ \cite{bolhuis2011relation,darve2013analysis},
\begin{equation}
\label{eq:to_rearrange}
    \frac{\pequil(x) \, \qa(x)}{\int_{B^c} \pequil(x') \, \qa(x')\,\mathrm{d}x'} = \frac{\patob(x)}{\int_{B^c} \patob(x')\,\mathrm{d}x'},
    \qquad x \notin B.
\end{equation}
The two sides are normalized by integrating over all points excluding those in $B$,
because the normalized $A \to B$ NESS density includes the sink endpoint in $B$, while $\qa(x)$ vanishes on $B$.
Rearranging \eqref{eq:to_rearrange} gives the likelihood ratio formula
\begin{equation}
\label{eq:likelihood_ratio}
    \qa(x) = \ptotatob\frac{\patob(x)}{\pequil(x)},
    \qquad x \notin B,
\end{equation}
where we have introduced a normalizing constant
\begin{equation}
    \ptotatob = \frac{\int_{B_c} \pequil(x') \qa(x')\,\mathrm{d}x'}{\int_{B^c} \patob(x')\,\mathrm{d}x'}.
\end{equation}
Because $\qa(x) = 1$ on $A$, integrating \eqref{eq:likelihood_ratio} over $A$ allows us to re-express the normalizing constant as
\begin{equation}
    \ptotatob = \frac{\int_A \pequil(x)\,\mathrm{d}x}{\int_A \patob(x)\,\mathrm{d}x}.
\end{equation}
Equivalently, wherever the densities are positive,
$\pequil(x)/\patob(x)=\ptotatob$ for $x \in A$.
These relations are derived using detailed balance.
Under momentum-flip detailed balance, a corresponding statement involves the committor $q^A$ evaluated for the momentum flipped phase-space point $x$.

To conclude, we emphasize the following key point that will allow us to estimate the coarse-grained committor without bias.
\key{
On the appropriate restricted domain, the committor is proportional to the likelihood ratio between the nonequilibrium and equilibrium steady state densities  \cite{darve2013analysis}.
}

\subsection{Mean first-passage times} \label{sec:mfpt-micro}

Another central quantity representing kinetics is the $A \to B$ mean first-passage time (MFPT), denoted by $\tlocatob$.
The $A \to B$ MFPT considered here is the mean duration of an $A \to B$ trajectory segment in the equilibrium ensemble, measured from its first sampled point in $A$ through its first sampled point in $B$.
Equivalently, the MFPT is the mean length of such a segment in the $A \to B$ nonequilibrium ensemble, which simply skips over any $B \to A$ trajectory segments by a lag-time step returning the trajectory from $B$ directly to $A$ (Fig.\ \ref{fig:long-traj}).

We now employ a variation of the Kac return-time argument \cite[Sec.~1.7]{norris1998markov} to derive a simple expression for the MFPT, assuming that the recycled process is stationary and ergodic.
In the $A \to B$ nonequilibrium ensemble, we observe that
each return cycle to $B$ consists of one lag-time step that recycles the trajectory from $B$ to $A$, followed by an $A \to B$ first-passage segment (Fig.\ \ref{fig:long-traj}), so the average time for a return cycle is
\begin{equation}
	\tau + \tlocatob.
\end{equation}
Because the recycled chain occupies $B$ at exactly one discrete observation per return cycle, the fraction of time spent visiting $B$ in the $A \to B$ nonequilibrium ensemble equals the lag time $\tau$ divided by the average time for a return cycle,
\begin{equation}
\label{eq:step1}
    \text{fraction of time visiting $B$} = \frac{\tau}{\tau + \tlocatob}.
\end{equation}
Yet the fraction of time spent visiting $B$ is also expressed by the $A \to B$ steady-state density of $B$, which is
\begin{equation}
\label{eq:step2}
    \text{fraction of time visiting $B$} = \int_B \patob(x)\,\mathrm{d}x
\end{equation}
By combining \eqref{eq:step1} with \eqref{eq:step2}, we obtain the expression for the MFPT,
\begin{equation}
\label{eq:hill}
    \tlocatob = \frac{\tau}{\int_B \patob(x)\,\mathrm{d}x} - \tau.
\end{equation}
This identity expresses the MFPT in terms of the stationary probability of the sink state in the source-sink steady state ensemble.



We can recast the expression \eqref{eq:hill} to make a connection to a rate constant.  We will use the indicator function for state $B$,
\begin{equation}
    \ind_B(x) = \begin{cases}
        1, & x \in B, \\
        0, & x \notin B.
    \end{cases}
\end{equation}
We apply the relation 
$\int \cdots \mathrm{d}x = \int_B \cdots \mathrm{d}x + \int_{B^c} \cdots \mathrm{d}x$ and the stationary condition $\int \patob(x') \,\ptauatob(x | x') \, \mathrm{d}x' = \patob(x)$ 
to rearrange \eqref{eq:hill} as
\begin{equation}
\begin{aligned}
    \frac{\tau}{\tlocatob}
    &= \frac{\int_B \patob(x)\,\mathrm{d}x}{\int_{B^c} \patob(x) \, \mathrm{d}x} \\
    &= \frac{\int \int \patob(x') \, \ptauatob(x | x')\,\ind_B(x) \, \mathrm{d}x \, \mathrm{d}x'}{\int_{B^c} \patob(x) \, \mathrm{d}x}.
    \label{hill-discrete}
\end{aligned}
\end{equation}
The numerator is the probability of starting outside of $B$ and entering $B$ on the next lag-time step, while the denominator is the probability of starting outside of $B$ in the $A \to B$ nonequilibrium ensemble.
Dividing by $\tau$,
we obtain the conditional entry probability into $B$ per unit time, which is a viable definition for the rate constant $\kab$ \cite{zuckerman2010book}.

Eq.~\eqref{eq:hill} is a discrete-time version of the Hill relation \cite{hill2013free}, which establishes the reciprocal flux relation in the continuous-time limit $\tau \to 0$.
Readers can also consult a related derivation for the continuous-time case \cite{zuckerman2015hill-cont,zuckerman2021gentle}.
We emphasize the following key point.

\key{The mean first-passage time $\tlocatob$ is the inverse rate constant for an $A \to B$ process, and it can be computed from the $A \to B$ NESS probability of the sink state,
\begin{equation}
    \tlocatob
    =
    \frac{\tau}{\int_B \patob(x)\,\mathrm{d}x}
    - \tau .
\end{equation}}

Last, we acknowledge a nuance in the MFPT definition provided here.
By our definition, the MFPT generally increases with the lag time, because coarser time sampling delays the first observed entry to the target state \cite{suarez2021markov-can-not}. 
Choosing an appropriate lag time is therefore a practical modeling choice.
The goal of capturing the \emph{first} passage into the target state suggests using the shortest possible lag.
However, unless state boundaries coincide with deep metastable basins, trajectories may rapidly recross the boundary, as emphasized in the transition state theory literature \cite{chandler1987introduction}.
The recrossing issue motivates using a longer lag time to measure transitions that remain in the target state rather than transient boundary crossings.
Different lag times define modified kinetic observables, so MSM estimates and reference simulations should always use the same convention \cite{suarez2021markov-can-not}.

\section{Unbiased coarse-grained observables from MSMs}
\label{sec:exact}

Markov state models (MSMs) are simplified models for the transitions among a set of clusters or ``states'', defined as non-overlapping regions fully covering the phase space.
There has been extensive discussion of the clustering methods for defining MSM states \cite{chodera-noe2014markov-review,pande2016optimized,vandeneijnden2016optimized}, so this perspective focuses on a different issue of unbiased observables for a given set of clusters.

Given the ideal trajectory data, we can define an MSM through a transition matrix $\tx = \tx(\tau)$, where each element $i,j$ describes the probability of observing the system in state $j$ following a time lag of $\tau$ after the system was observed in state $i$. This matrix can be written
\begin{equation}
    \txij = \int \pinit(x' | i) \int \ptau^{\textrm{X}}(x | x') \, \ind_j(x)\,\mathrm{d}x\,\mathrm{d}x'.
\label{tij-from-init}     
\end{equation}
Here, $\pinit(x' | i)$ denotes an initial distribution of phase-space points within each state $i$, and
$\ind_j$ is the indicator function on the state $j$, which satisfies
\begin{equation}
	\ind_j(x) = \begin{cases}
		1, & x \in j, \\
		0, & x \notin j.
	\end{cases}
\end{equation}
Last, $X$ represents the ensemble of interest, which could be the equilibrium ensemble or $A \to B$ nonequilibrium ensemble.
The matrix elements are fully determined by the distribution of initial points, along with the dynamics and boundary conditions embodied in $\ptau^X$.

Although MSMs are approximations obtained by coarse-graining in both space and time, 
remarkably they can express several coarse-grained observables at any fixed lag-time without bias, provided the MSM matrix elements are constructed from the correct stationary within-cluster distributions.
\begin{itemize}
\item[(i)] Equilibrium stationary probabilities are obtained from the equilibrium MSM $\tequil$.
\item[(ii)] $A \to B$ stationary probabilities are obtained from the $A \to B$ nonequilibrium MSM $\tatob$.
\item[(iii)] The $A \to B$ MFPT is also obtained from $\tatob$.
\item[(iv)] The $A$-committor is obtained from a combination of $\tequil$ and $\tatob$.
\end{itemize}
The rest of this section will describe these procedures in detail.


\subsection{Stationary probabilities}

To evaluate unbiased equilibrium stationary probabilities,
we can generate the equilibrium MSM with the following matrix elements.
\begin{equation}
	\tequilij = \int \pequil(x'|i) \int \ptau(x | x') \, \ind_j(x)\, \mathrm{d}x \, \mathrm{d}x'.
\label{tij-equil}
\end{equation}
Here, $\pequil(x | i)$ denotes the equilibrium density conditioned on cluster $i$,
\begin{equation}
	\pequil(x | i) = \frac{\ind_i(x) \, \pequil(x)}{ \ppequil_i }.
    \label{p-equil-i}
\end{equation}
Its normalization factor is the equilibrium stationary probability for cluster $i$
\begin{equation}
	\ppequil_i = \int I_i(x)\, \pequil(x)\,\mathrm{d}x.
\end{equation}
We can check that when we solve for the stationary distribution of the transition matrix $\tequil$ defined in \eqref{tij-equil}, we recover the equilibrium stationary probabilities $\ppequil_i$, regardless of the lag time $\tau$.
\begin{equation}
\label{equil-check}
\begin{aligned}
	\sum_i \ppequil_i \, \tequilij
	&= \int \pequil(x') \int \ptau(x | x') \, \ind_j(x) \, \mathrm{d}x \, \mathrm{d}x' \\
     &= \int I_j(x)\, \pequil(x)\,\mathrm{d}x = \ppequil_j  \hspace{0.5cm} \hbox{for any }\tau.
\end{aligned}
\end{equation}
The first equality holds because
\begin{equation}
    \sum_i \ppequil_i \pequil(x'|i)=\pequil(x'),
\end{equation}
which follows from the definition \eqref{p-equil-i} and the non-overlapping clusters covering phase space.
The second equality of \eqref{equil-check} holds because of the stationary condition, $\int \pequil(x') \, \ptau(x | x') \, \mathrm{d}x' = \pequil(x)$.

In exact analogy to equilibrium stationary probabilities, we can recover the nonequilibrium stationary probabilities from an approriate MSM model.
We form the $A \to B$ NESS MSM with elements given by
\begin{equation}
	\tatobij = \int \patob(x'|i) \int \ptau^{A \to B}(x | x') \ind_j(x) \, \mathrm{d}x \, \mathrm{d}x'.
\label{tij-atob}     
\end{equation}
The initial points in each cluster $i$ are sampled from the $A \to B$ NESS density conditioned on cluster $i$,
\begin{equation}
	\patob(x | i) = \frac{\ind_i(x) \, \patob(x)}{\ppatob_i}.
\label{patobi}
\end{equation}
Its normalization factor is
\begin{equation}
	\ppatob_i = \int I_i(x)\,\patob(x)\,\mathrm{d}x.
\label{ppatob-i}
\end{equation}
Similar to \eqref{equil-check}, the stationary vector of the transition matrix $\tatobij$ yields the probabilities $\ppatob_i$,
\begin{equation}
    \sum_i \ppatob_i \, \tatobij = \ppatob_j.
    \label{ss-atob}
\end{equation}

This approach for calculating stationary probabilities may seem circular, but it underpins important practical algorithms as we will see in Sec.~\ref{sec:algo}.
The failure to reckon with the proper construction of the equilibrium or nonequilibrium MSM by using \emph{ad hoc} initial distributions will inevitably lead to biased results, not only for stationary probabilities but for any observables computed from these matrices.

\key{If matrix elements are constructed based on an initial distribution of phase points exactly corresponding to stationarity, then the stationary solution of the MSM equals the coarse-grained stationary distribution.
If different within-cluster initial distributions are used, this equality is no longer guaranteed, and stationary observables can be biased.}

\subsection{Cluster-to-cluster fluxes}

In a discretized description of phase space, a fundamental mechanistic quantity is the flux $\fij^X = \fij^X(\tau)$, which describes the directed probability movement from $i \to j$ occurring over lag time $\tau$.
The flux is defined as
\begin{equation}
	    \fij^X = \int I_i(x')\, \pi^X(x') \int \ptau^X(x | x') \,
        \ind_j(x) \,\mathrm{d}x \, \mathrm{d}x'.
\end{equation}
Here, as usual, $X$ represents either the equilibrium ensemble or the $A \to B$ nonequilibrium ensemble.
In both cases, the flux can be computed from the matched MSM matrix and its stationary probabilities, via
\begin{equation}
    \fij^X = \pp^X_i \, \tij^X.
\label{flux-genl}
\end{equation}
This MSM calculation is unbiased provided the correct coarse-grained stationary density $\pi_i^X$ and transition matrix $\tij^X$ are used.

\key{
The cluster-to-cluster flux $\fij^X$ represents the exact amount of probability moving from cluster $i$ to $j$ in ensemble $X$,
and it can be computed from the matched MSM and its stationary probabilities, using $\fij^X = \pp^X_i \, \tij^X$.}  

\subsection{Mean first-passage time}\label{sec:mfpt-coarse}

The $A \to B$ mean first-passage time quantifies the inverse rate constant for an $A \to B$ process, as discussed in Sec.~\ref{sec:mfpt-micro}.
Here we show how to compute the $A \to B$ MFPT without bias from an MSM model.
Applying the discrete-time Hill relation \eqref{eq:hill}, 
\begin{equation}
    \tlocatob
    =
    \frac{\tau}{\int_B \patob(x)\,\mathrm{d}x}
    - \tau .
\end{equation}
Here and throughout, we assume that the union of one or more MSM clusters equals the target state $B$.
We can rewrite the MFPT expression using the coarse-grained $A \to B$ stationary probability for the $B$ state,
\begin{equation}\label{eq:Hill_discrete}
    \tlocatob
    =
    \frac{\tau}{\sum_{i \in B} \patob_i}
    - \tau.
\end{equation}

We emphasize the following key point.

\key{
The exact $A \to B$ MFPT for lag time $\tau$ is determined by the $A\to B$ NESS stationary probability of the $B$ state.
\begin{equation}
    \tlocatob = \frac{\tau}{\sum_{i \in B} \patob_i} - \tau .
\end{equation}
This identity is unbiased when $\patob_b$ is obtained from the correctly constructed $A\to B$ NESS transition matrix because the stationary probability itself is unbiased.}

\subsection{Committor}
\label{sec:comm-discrete}

To define the proper coarse-graining for the committor, we are guided by the long trajectory perspective.
In a long trajectory, points in a given coarse cluster will be equilibrium-distributed, so we define
the coarse-grained committor $\qa_i$ as
\begin{equation}
\label{eq:committor_defn}
\begin{aligned}
	\qa_i &= \int \pequil(x | i) \, \qa(x) \, \mathrm{d}x \\
	&= \frac{\int I_i(x)\, \pequil(x) \, \qa(x) \, \mathrm{d}x}{\int I_i(x)\, \pequil(x) \, \mathrm{d}x}.
\end{aligned}
\end{equation}
In the second line, we have used the explicit formula for the equilibrium density conditioned on cluster $i$.
We can generate an unbiased estimate for the coarse-grained committor by using the likelihood ratio formula which was derived in Sec.~\ref{sec:continuous_committor},
\begin{equation}
\label{eq:substitute_me}
    \qa(x) = \ptotatob\frac{\patob(x)}{\pequil(x)},
    \qquad x \notin B.
\end{equation}
For any cluster $i$ outside of the sink state $B$, substituting \eqref{eq:substitute_me} into \eqref{eq:committor_defn} gives
\begin{equation}
	\qa_i  = \ptotatob \frac{\patob_i}{\pequil_i}.
\end{equation}
For any cluster $a$ in the source state $A$, the condition $\qa_a = 1$ leads to an explicit formula for the normalizing constant,
\begin{equation}
    \ptotatob = \frac{\pequil_a}{\patob_a}.
\end{equation}
We can use these formulas to obtained unbiased estimates of the coarse-grained committor and we close with the following key point.

\key{In analogy to the continuous-state case, the coarse-grained committor is the normalized ratio of steady-state populations between the $A \to B$ NESS and equilibrium.}

\subsection{Traditional MSM observables are biased}
\label{sec:first-algo}

The traditional MSM formulas for the mean first-passage time and committor are exact for a Markov chain on the MSM states, but they are generally biased when the MSM states are coarse-grained regions of a continuous-state molecular system.
These formulas are inconsistent with the unbiased formulas in Secs.~\ref{sec:mfpt-coarse} and \ref{sec:comm-discrete}.

The traditional MSM formula for the $A \to B$ MFPT is based on a first-step relation that holds under Markovianity.
If the lag-$\tau$ dynamics are fully described by a Markov transition matrix with elements $\tij$, then the local MFPT to $B$ is given by the following formula \cite{privault2018understanding}.
\begin{equation}
    \label{first-step-mfpt}
    \begin{cases}
    \tloc_i^B = \sum_j \tij \, \tloc_j^B + \tau,
    & i \notin B, \\
    \tloc_i^B = 0, & i \in B,
    \end{cases}
    \hspace{1cm} \mbox{(Unbiased assuming Markov property).}
\end{equation}
This means the local MFPT from a given state $i \notin B$ is the weighted average of the MFPTs among the subsequently visited states, plus one lag time.
A global $A \to B$ MFPT is then obtained by averaging these local MFPTs over the first entry distribution on $A$.

In direct analogy, the traditional MSM formula for the coarse-grained committor is also given by a first-step relation that holds under the Markov property.
The Markov property implies that the coarse-grained committor satisfies
\begin{equation}
    \begin{cases}
    \qa_i = \sum_j \tij \, \qa_j,
    & i \notin A, \, i \notin B, \\
    \qa_i = 1, & i \in A, \\
    \qa_i = 0, & i \in B,
    \end{cases}
    \hspace{1cm} \mbox{(Unbiased assuming Markov property).}
    \label{first-step-comm}     
\end{equation}
In other words, the probability to reach state $A$ before $B$ from a state $i$ outside of $A$ or $B$ must be preserved when averaging over all possible one-step outcomes.

Both the relations \eqref{first-step-mfpt} and \eqref{first-step-comm} are systematically biased, regardless of how much trajectory data is used to estimate the transition matrix, because the projected dynamics on coarse-grained molecular states are generally not Markovian.
The average observable in the downstream state $j$ does not necessarily describe all future behavior from $j$, because different phase-space regions within $j$ can lead to different outcomes.
Secs.~\ref{sec:limitations_MFPT} and \ref{sec:comm-error} will analyze the error in the first-step relations and provide formulas indicating that this error might be large.

\section{Trajectory reweighting}
\label{sec:algo}

The preceding section has explained how to use MSMs to compute unbiased observables, given infinite data that is distributed according to the appropriate stationary distribution.
However, in practice we never have perfect data, so this section reviews ways to reweight trajectory data to construct more accurate MSMs.

\subsection{Baseline method with simple counts}

First consider a baseline method that constructs the MSM transition matrix based on simple counts.
Suppose we have access to a data set of trajectories $x_k$, for $k = 1, 2, \ldots, K$.
Each trajectory is observed for $N_k$ time steps,
\begin{equation}
	x_k(0), x_k(\tau), \ldots, x_k(N_k \tau).
\end{equation}
For each $k = 1, 2, \ldots, K$, we define a transition count matrix $C^{(k)}$ with elements
\begin{equation}
	\cij^{(k)} = \sum_{n=1}^{N_k} \ind_i(x_k(n\tau - \tau)) \, \ind_j(x_k(n\tau )).
\end{equation}
We also define a cluster occupancy vector with elements determined from $C^{(k)}$ as follows.
\begin{equation}
\begin{aligned}
	\ci^{(k)} &= \sum_j C^{(k)}_{ij} \\
	&= \sum_{n=0}^{N_k-1} \ind_i(x_k(n\tau)).
\end{aligned}
\end{equation}
The simplest way we can compute entries of an MSM is by dividing the sum of the transition counts over all the trajectories by the sum of the cluster occupancies,
\begin{equation}
\label{tij-counts}
	\tij = \frac{\sum_{k=1}^K \cij^{(k)}}{\sum_{k=1}^K \ci^{(k)}}.
\end{equation}
In this simple-counts approach, each observed lag-time transition is weighted equally, so longer trajectories generally make the largest contributions to the sums.

In several idealized settings, the simple-counts approach can produce asymptotically unbiased MSMs as the number of trajectories $K$ goes to infinity.
\begin{itemize}
\item The simple-counts estimator is unbiased for the equilibrium MSM when the trajectory start points are sampled from the equilibrium density $\pequil$.
\item The estimator is unbiased for the \(A \to B\) nonequilibrium MSM when the trajectory start points are sampled from the \(A \to B\) nonequilibrium density $\patob$ and the transitions follow the $A \to B$ recycling dynamics.
\item Last, the simple-counts estimator is unbiased for \(\tequil\) or \(\tatob\) when the start points of the counted transitions are sampled from the correct within-cluster distributions \(\pequil(x | i)\) or \(\patob(x | i)\), and each trajectory consists of a single lag-time transition according to $\ptau$ or $\ptau^{A \to B}$.
\end{itemize}

In contrast, the simple-counts estimator becomes biased when initial points are not drawn from the target stationary density and many transitions are counted.
The downstream transition origins are distributed according to the relaxation of the initial ensemble, rather than the desired stationary density.

In the typical MD workflow, researchers have access to one or more long trajectories, some of which pass through the $A$ and $B$ states of interest.
Researchers then apply simple counts to the full data set of trajectories to approximate the matrix $\tequil$.
As a follow-up step, some researchers subset the data to include only trajectory segments whose most recent visit was to $A$, and they terminate each segment upon hitting $B$.
They apply simple counts to the subsetted data and enforce the recycling dynamics from the sink $B$ to the source $A$, resulting in a ``history-augmented MSM'' (haMSM)  \cite{suarez2016nonmarkov,copperman2020hamsm,suarez2021markov-can-not} that approximates the matrix $\tatob$.

While the simple counts approach is powerful in idealized settings, it has limitations in typical MD applications.
First, because full trajectories are counted, the start points must be sampled from the target stationary density for unbiased estimation.
However, in practice the start points are only partially equilibrated, so finite trajectories retain initialization bias, especially along the slow relaxation modes.
Second, the haMSM approach requires tracing trajectories backward to the last visit to $A$ or $B$, which is not always possible and can reduce the amount of usable data.
To address these issues, the rest of this section will discuss alternative strategies that reweight the available trajectory data so it better matches the target equilibrium or nonequilibrium distribution.

\subsection{Trajectory reweighting based on static clustering}

The deviation from stationarity in typical data sets \cite{vitalis2019removal} has motivated efforts to remove initial-distribution bias from MSM models.
Wan \& Voelz \cite{voelz2020reweight} developed the first proposal for reweighting trajectories to improve consistency with a target stationary distribution $\pp$.
They were inspired by the ideal weighted sampling approach, in which each trajectory $k = 1, \ldots, K$ receives a weight $w_k$ so that the weighted distribution of start points equals the target distribution $\pp$.
Given an increasing quantity of perfectly weighted data, the weights and trajectory start points yield
\begin{equation}
\label{eq:exact}
	\sum_{k=1}^K w_k \ind_i(x_k(0)) = \pp_i,
    \qquad \text{for each cluster $i$}.
\end{equation}
In practice, we do not know the exact stationary probabilities $\pp_i$, but we can check the related consistency condition
\begin{equation}
\label{eq:to_approximate}
    \sum_{k=1}^K w_k \ind_i(x_k(0)) = \pp_i^{\rm MSM},
    \qquad \text{for each cluster $i$}.
\end{equation}
Here $\pp_i^{\rm MSM}$ is the stationary probability for cluster $i$ calculated from the MSM model.
Violations of \eqref{eq:to_approximate} may indicate
initial-distribution bias and can contribute to biased MSM estimates.

To improve correspondence with the condition \eqref{eq:to_approximate}, Wan \& Voelz first calculated the stationary solution $\pi^{\rm MSM}$ for an initial, standard MSM.
Then they assigned a uniform weight to each trajectory beginning in cluster $i$,
\begin{equation}
\label{eq:wt-voelz}
	w_k = \frac{\pp_i^{\rm MSM}}{\sum_{\ell=1}^K \ind_i(x_\ell(0))},
	\qquad \text{if } x_k(0) \in i.
\end{equation}
Last, they redefined the MSM transition matrix via
\begin{equation}
    \tij = \frac{\sum_{k=1}^K w_k \, \cijk}{\sum_{k=1}^K w_k \, \cik}.
    \label{tij-voelz}
\end{equation}
Wan \& Voelz initially applied their trajectory reweighting approach in the equilibrium setting \cite{voelz2020reweight}.
However, the same approach can be used to approximate the nonequilibrium $A \to B$ transition matrix when the trajectories follow the appropriate source-sink recycling dynamics.

As a limitation, however, the original Wan \& Voelz algorithm was not fully self-consistent, since the matrix elements defined in \eqref{tij-voelz} generally result in a new stationary solution that is no longer consistent with the assigned trajectory weights.
Therefore, Russo et al. \cite{russo2020iterative} proposed an iterative version of the algorithm
that repeatedly updates the weights defined in \eqref{eq:wt-voelz} after computing each new stationary solution from the weighted matrix elements \eqref{tij-voelz}.
At convergence, the iterated approach produces weights and stationary probabilities that satisfy \eqref{eq:to_approximate} for the chosen static clustering.

These early trajectory reweighting algorithms \cite{voelz2020reweight,russo2020iterative} were primarily intended for long trajectories.
For a single-lag trajectory, the trajectory weight multiplies both the numerator and denominator of one row of the transition matrix and therefore cancels.
Therefore, trajectory reweighting results in no change to the computed MSM matrix entries.
Reweighting only begins to affect the MSM matrix as the number of lag-time steps increases, since it changes the computed matrix entries for the downstream bins visited by the trajectories.
Given finite trajectory length, reweighting based on static bins cannot fully correct the distribution of trajectory start points within each cluster.

\key{Reweighting based on static clusters is insufficient for unbiased local sampling with finite trajectories, especially when the trajectories are short.}

\subsection{Trajectory reweighting with random clusters}

The early algorithms for trajectory reweighting \cite{voelz2020reweight,russo2020iterative} were limited due to their reliance on static clusters.
In contrast, the motivating idea behind the randomized iterative reweighting (RiteWeight) approach \cite{kania2026randomized} is to randomly change the cluster definition at each iteration and then identify a new stationary solution of the MSM transition matrix
\begin{equation}\label{tij-riteweight}
    \tij = \frac{\sum_{k=1}^K w_k \, \cijk}{\sum_{k=1}^K w_k \, \cik}.
\end{equation}
The algorithm then enforces the consistency condition \eqref{eq:to_approximate} by updating the weight of each trajectory beginning in cluster $i$ according to
\begin{equation}
\label{reweight-eq}
	\wtnew_k = 
	\frac{\pp_i^{\rm MSM}}{\sum_{\ell = 1}^K \ind_i(x_{\ell}(0)) w_{\ell}} w_k,
    \qquad \text{if } x_k(0) \in i.
\end{equation}
Here $w_k$ is the previous weight of the $k$th trajectory, $\wtnew_k$ is the new weight, and $\pp_i^{\rm MSM}$ is the current iteration's estimate of the stationary probability for cluster $i$, based on $T$.
The update rule \eqref{reweight-eq} preserves the relative weights of trajectories starting in each cluster during a single iteration, but repeated re-clustering allows trajectories that were previously grouped together to receive different relative weights in later iterations.
In this way, the algorithm iteratively improves the consistency between the trajectory weights and the stationary vectors of the weighted MSMs.
We can apply this approach to either an equilibrium distribution or an $A \to B$ steady-state distribution, by imposing the appropriate boundary conditions at the sink and removing any segments violating the boundary conditions.


The iterative algorithm described above can be applied to trajectories of any finite length.
However, for RiteWeight, long trajectories
are divided into segments containing one lag interval, allowing each counted transition origin to receive its own weight rather than tying its weight to all transitions in the parent trajectory.
This increases the flexibility of the reweighting, although it does
not add new dynamical information.

RiteWeight makes the greatest impact when the empirical distribution of transition origins is far from stationary under the target dynamics.
This occurs, for example, when a large number of short trajectories begin from a strongly biased distribution, or when data generated in one ensemble is reweighted to a different target ensemble, such as an $A\to B$ nonequilibrium steady state \cite{kania2026randomized}.
By contrast, RiteWeight makes little correction when the single-lag segments from long equilibrium trajectories are used to estimate equilibrium.
With long equilibrium tajectories, the consistency condition \eqref{eq:to_approximate} is automatically nearly satisfied with uniform weights, so there is little empirical inconsistency for RiteWeight to
correct \cite{zuckerman2026msmcounts}.

The  recent study \cite{kania2026randomized} shows that RiteWeight is asymptotically unbiased under idealized settings, and it empirically reduces the bias in equilibrium and nonequilibrium MSMs.
Further analysis and best practices for RiteWeight are actively under development, and practical applications often use time-averaging \cite{kania2026randomized} or
smoothing \cite{otten2026regularized} of the weights produced by the algorithm.
There is also an extended version of RiteWeight that can analyze trajectory segments generated from biased transition densities \cite{aristoff2026briteweight}.

\key{By repeatedly changing the clustering, RiteWeight approaches
self-consistency across a rich family of coarse-grainings rather than for
one fixed state decomposition.
The method can use nearly all available trajectory data of any length, and it can target either equilibrium or $A \to B$ steady state.}

\section{Theoretical error analysis of standard MSMs}
\label{sec:error-theory}

This section analyzes the error in traditional MSM observables due to the failure of Markovianity or incorrect local sampling within clusters.
We will derive exact error expressions for stationary probabilities, mean first-passage times, and committors.

\subsection{Analysis of stationary probabilities} \label{sec:distrlimits}

First we analyze the error when approximating a coarse-grained stationary probability vector, such as $\pequil$ or $\patob$.
The size of the error depends on how well the MSM transition matrix $\tmat$ approximates the target transition matrix $\tequil$ or $\tatob$.
The following proposition makes this connection precise.

\begin{proposition}[Errors in stationary probabilities] \label{prop:stationary}
	Let $\px$ be the stationary vector for a target MSM transition matrix $\tx$, where $X$ could indicate the equilibrium or nonequilibrium $A \rightarrow B$ ensemble, and let $\pp$ be the stationary vector for a computed MSM transition matrix $\tmat$.
    We assume $\tx$ is irreducible and aperiodic.
    Then the errors in the stationary probabilities are given by
	\begin{equation}
    \label{eq:stationary_error}
		\pp^\top - (\px)^\top = \pp^\top (\tmat - \tx) M^X
	\end{equation}
    where $M^X$ denotes the matrix
    \begin{equation}
        M^X = \left[ I - \tx + 1 (\px)^\top \right]^{-1}.
    \end{equation}
\end{proposition}
\begin{proof}
First make the calculation
\begin{equation}
\begin{aligned}
    (\pp - \px)^\top (I - \tx + 1(\px)^\top) 
    &= (\pp-\px)^\top(I-\tx) \\
    &= \pp^\top-\pp^\top\tx \\
    &= \pp^\top(\tmat-\tx).
\end{aligned}
\end{equation}
where the first equality uses
$(\pp-\px)^\top 1 = 0$, and the remaining equalities use
$(\px)^\top\tx=(\px)^\top$ and
$\pp^\top\tmat=\pp^\top$.
Multiply the left- and right-hand sides by $M^X$ to confirm
\begin{equation}\label{eq:pidifference1}
    \pp^\top - (\px)^\top = \pp^\top (\tmat - \tx) M^X.
\end{equation}
This completes the proof.
\end{proof}

In the error expression \eqref{eq:stationary_error}, the stationary probabilities are generally incorrect if the computed matrix $\tmat$ fails to match the target matrix $\tx$, which might result from misdistributed initial data.
The errors $\tmat - \tx$ are then inflated by a fundamental matrix
\begin{equation}
\begin{aligned}
    M^X &= \left[ I - \tx + 1 (\px)^\top \right]^{-1} \\
    &= I + \sum_{s=1}^{\infty} \left[(\tx)^s - 1 (\px)^\top\right]. \label{eq:neumann}
\end{aligned}
\end{equation}
The $i, j$ element of this matrix quantifies the expected number of visits to cluster $j$ in excess of an equilibrium value when a trajectory is started from cluster $i$, for the Markov model defined by $\tx$.
This implies that errors in the transition matrix are most significant when they inject probability into slowly relaxing modes or metastable regions.

As a simple corollary, we can use
Prop.~\ref{prop:stationary} to analyze the error in the
cluster-to-cluster flux $\fij^X = \px_i \tx_{ij}$ when the computed flux is
\begin{equation}
    \fij = \pp_i \tmat_{ij}.
\end{equation}
Element-wise, the error is
\begin{equation}
    \Phi_{ij}-\Phi^X_{ij}
    =
    \pp_i(\tij-\tx_{ij})
    +
    (\pp_i-\px_i)\tx_{ij}.
    \label{eq:flux-error-elementwise}
\end{equation}
The first term is a direct contribution from the error in the computed $i \to j$ transition probability.
The second term results from the induced error in the stationary probability for cluster $i$.
By Prop.~\ref{prop:stationary}, this
stationary-probability error can be amplified by the slow relaxation
timescales encoded in the fundamental matrix $M^X$.

\subsection{Analysis of mean first-passage times} \label{sec:limitations_MFPT}

This section analyzes the error in the $A \to B$ MFPT which is estimated using the first-step relation discussed in Sec.~\ref{sec:first-algo}.
The computed MFPT is correct when the first-step relation is applied with the $A \to B$ nonequilibrium transition matrix $\tatob$.
It is generally incorrect when it is applied to any other matrix, the resulting MFPT estimate, and Prop.~\ref{prop:mfpt} quantifies the size of the error.
This analysis reveals a potentially important source of bias with the traditional MSM approach that estimates the MFPT from an approximation of the \emph{equilibrium} transition matrix $\tequil$, instead of the $A \to B$ NESS matrix.

\begin{proposition}[Error in mean first-passage time] \label{prop:mfpt}
Let $\tau^{A \rightarrow B}$ be the exact $A \to B$ MFPT, and let $T$ be a computed MSM transition matrix, where state $A$ is represented by a single cluster $a$.
Let $\tloc_a$ be the estimated $A \to B$ MFPT from the first-step relation
\begin{equation}
    \begin{cases}
        \tloc_i = \sum_j \tmat_{ij} \tloc_j + \tau, & i \notin B, \\
        \tloc_i = 0, & i \in B.
    \end{cases}
\end{equation}
Then the error in the $A \to B$ MFPT is given by
\begin{equation}
    \tloc_a - \tlocatob
    = e_a^\top (I - \tilde{\tmat}^{A \to B})^{-1} (\tilde{\tmat} - \tilde{\tmat}^{A \to B}) \tilde{\tloc}.
\end{equation}
Here, $e_a$ is the standard basis vector associated with cluster $a$. $\tilde{\tmat}^{A \to B}$, $\tilde{\tmat}$, and $\tilde{\tloc}$ denote submatrices or subvectors of $\tmat$, $\tmat^{A \to B}$, and $m$ after removing all entries corresponding to clusters $i \in B$.
\end{proposition}
\begin{proof}
It was shown in Sec.~\ref{sec:mfpt-coarse} that the $A \to B$ MFPT for the continuous-state system satisfies
\begin{equation}
    \tlocatob = \frac{\tau}{\sum_{i \in B} \patob_i}-\tau.
\end{equation}
By the Kac return-time argument from Sec.~\ref{sec:mfpt-micro}, the $A \to B$ MFPT for the Markov chain with transition matrix $\tatob$ also equals $\tau / \sum_{i \in B} \patob_i - \tau$.
Therefore, the $A \to B$ MFPT equals the quantity $\tloc_i^B$ determined by the first-step relation
\begin{equation}
    \begin{cases}
        \tloc_i^B = \sum_j \tatob_{ij} \tloc_j^B + \tau, & i \notin B, \\
        \tloc_i^B = 0, & i \in B.
    \end{cases}
\end{equation}
We combine the first-step relations for $\tmat$ and $\tatob$ to yield a matrix equation
\begin{equation}
    \tilde{\tloc} - \tilde{\tloc}^B  
    = \tilde{\tmat} \tilde{\tloc} - \tilde{\tmat}^{A \to B} \tilde{\tloc}^B
\end{equation}
Rearranging the equation,
\begin{equation}
    (I - \tilde{\tmat}^{A \to B})(\tilde{\tloc} - \tilde{\tloc}^B)  
    = (\tilde{\tmat} - \tilde{\tmat}^{A \to B}) \tilde{\tloc}
\end{equation}
Multiplying the left- and right-hand sides by $(I - \tilde{\tmat}^{A \to B})^{-1}$,
\begin{equation}
    \tilde{\tloc} - \tilde{\tloc}^B
    = (I - \tilde{\tmat}^{A \to B})^{-1} (\tilde{\tmat} - \tilde{\tmat}^{A \to B}) \tilde{\tloc}.
\end{equation}
Last, we restrict to the $a$ state.
\begin{equation}
    \tloc_a - \tlocatob
    = e_a^\top (I - \tilde{\tmat}^{A \to B})^{-1} (\tilde{\tmat} - \tilde{\tmat}^{A \to B}) \tilde{\tloc}.
\end{equation}
This completes the proof.
\end{proof}

Prop.~\ref{prop:mfpt} shows how any discrepancies between $\tilde{\tmat}$ and $\tilde{\tmat}^{A \to B}$ are inflated by a fundamental matrix
\begin{equation}
    (I - \tilde{\tmat}^{A \to B})^{-1} = \sum_{s=0}^{\infty } (\tilde{\tmat}^{A \to B})^s.
\end{equation}
The $a,j$ element of this matrix quantifies the expected number of visits to cluster $j$ before hitting $B$ for the Markov model defined by $\tmat^{A \to B}$, when a trajectory is started from cluster $a$.
This implies that the errors in the transition matrix are most significant when they occur in origin states that are frequently visited, shift probability toward states with substantially different remaining MFPTs, or involve a combination of these two effects.

\subsection{Analysis of committors}
\label{sec:comm-error}

Last, we evaluate the error in the coarse-grained committor $\qa$.
The standard MSM approach estimates $\qa$ from a first-step relation involving a computed transition matrix $\tmat$.
However, this approach is generally biased unless
\begin{equation}
    \qa_i = \sum_j \tmat_{ij} \qa_j,
    \qquad i \notin A, \, i \notin B,
\end{equation}
which is only expected to hold if the coarse-grained dynamics satisfy the Markov property.
The following proposition quantifies the size of the bias.

\begin{proposition}[Error in committor] \label{prop:committor}
Let $T$ be a computed MSM transition matrix, and let $q$ be the estimated $A$-committor from the first-step relation
\begin{equation}
\begin{cases}
    q_i = \sum_j \tmat_{ij} q_j, & i \notin A, \, i \notin B, \\
    q_i = 0, & i \in B, \\
    q_i = 1, & i \in A.
\end{cases}
\end{equation}
Then the error in the $A$-committor is given by
\begin{equation}
\label{eq:committorbound_final}
    q^\circ - q^{A,\circ}
    = (I - \tmat^\circ)^{-1} d^{\circ},
\end{equation}
where 
\begin{equation}
d = \tmat \qa - \qa
\end{equation}
is the discrepancy vector.
Here, $q^\circ$, $q^{A, \circ}$, $\tmat^\circ$, and $d^\circ$ denote submatrices or subvectors of $q$, $\qa$, $\tmat$, and $d$ after removing all entries corresponding to states $i \in A$ or $i \in B$.
\end{proposition}
\begin{proof}
Introduce a vector $t$ with entries $t_i = \sum_{j \in A} T_{ij}$.
Then we can rewrite the first-step relation as a linear system
\begin{equation}
    q^\circ = \tmat^\circ q^\circ + t^\circ.
\end{equation}
Next use the definition of the discrepancy vector to check that
\begin{equation}
\begin{aligned}
    (I - \tmat^\circ)(q^\circ - q^{A, \circ}) 
    &= (\tmat^\circ q^{A, \circ} + t^\circ) - q^{A, \circ} \\ 
    &= d^\circ,
\end{aligned}
\end{equation}
Multiplying both sides by $(I - \tmat^\circ)^{-1}$,
\begin{equation}
    q^\circ - q^{A,\circ} = (I - \tmat^\circ)^{-1} d^\circ.
\end{equation}
This completes the proof.
\end{proof}

The error expression \eqref{eq:committorbound_final} shows how violations of the equality
$\qa_i = \sum_j \tmat_{ij} \qa_j$ are amplified by a fundamental matrix
\begin{equation}
    (I - \tmat^\circ)^{-1} = \sum_{s=0}^{\infty} (\tmat^\circ)^s.
\end{equation}
The $i,j$ element of this matrix quantifies the expected number of visits to cluster $j$ before hitting $A$ or $B$ for the Markov model defined by $\tmat$, when a trajectory is started from cluster $i$.
This implies that errors in the transition matrix are most impactful when they occur in origin states that are frequently visited, cause large committor discrepancies, or involve a combination of the two effects.

\section{Concluding discussion}

This perspective paper has examined the limitations of the standard MSM framework in which a single transition matrix is used to estimate many observables.
Even with infinite sampling, this approach can produce biased observables because different observables generally require transition matrices with different within-cluster sampling distributions.
In Sec.~\ref{sec:error-theory}, we quantified these biases through exact error expressions for stationary probabilities, mean first-passage times, and committors.
These results demonstrate how coarse-grained observables depend strongly on
the distribution of configurations sampled within each state.

It is useful to recall why constructing an exactly Markovian MSM is difficult.
Biomolecular energy landscapes contain large numbers of basins and barriers, as illustrated in Fig.~\ref{fig:landscape}.
Given the affordable amount of trajectory data, it is generally infeasible to construct a moderate-sized state decomposition that resolves every dynamically relevant basin and barrier.
Consequently, individual MSM states often contain configurations with substantially different future behavior, leading to the failure of the Markov property at practical lag times.
As one indication of this complexity, a classic study by Czerminski and Elber cataloged 138 local energy minima for a tripeptide \cite{elber1990rxnpaths};
later work has further documented the remarkable landscape complexity of seemingly simple systems \cite{wales2018landscapes,wales2019pathways}.
Overall, the notion of describing a protein-scale system with, say, 100 - 1,000 strictly Markovian states brings with it insurmountable challenges.

Rather than requiring a single Markovian transition matrix, we have presented a framework for computing unbiased coarse-grained observables from data.
Our framework uses an equilibrium transition matrix and an $A \to B$ source–sink transition matrix, each constructed using the corresponding stationary within-cluster sampling distribution.
Together, these objects yield unbiased estimators for equilibrium and nonequilibrium stationary probabilities, mean first-passage times, and committors in the infinite-data limit.

Our theoretical discussion has focused primarily on systematic bias in observables and has largely set aside uncertainty due to finite sampling.
Readers interested in finite-sample limitations of MSMs can consult \cite{scalco2011equilibrium,vitalis2019removal,grubmuller2023uncertainties}.
Nevertheless, Sec.~\ref{sec:algo} describes trajectory-reweighting methods that begin to connect the infinite-data theory with realistic
data regimes.
Because the proposed estimators can be incorporated into existing MSM workflows, and because these reweighting methods have shown
strong empirical performance, we expect the framework to provide a practical route toward reducing systematic MSM bias.
Understanding its performance under finite sampling, imperfect reweighting, and incomplete state-space coverage remains an important direction for future work.

\section*{Acknowledgments}
Jeremy Copperman, John Russo, and Gideon Simpson played a key role in developing the perspective presented here.
We also thank Alex Dickson, Sagar Kania, and Edward Lyman for useful discussions of MSM properties. RJW acknowledges support from a 2026--2027 Hellman Fellowship. DMZ acknowledges support from NIH Grant GM115805.


\printbibliography

\end{document}